\documentclass{article}
\pdfoutput=1
\usepackage{algorithm}
\usepackage{algpseudocode}
\usepackage{amsthm}
\usepackage{amsmath}
\usepackage{amssymb}
\usepackage{asymptote}
\usepackage{color}
\usepackage{float}
\usepackage{enumitem}
\usepackage{cancel}
\usepackage[hidelinks]{hyperref}

\newtheorem{lemma}{Lemma}
\newtheorem{corollary}{Corollary}
\newtheorem{definition}{Definition}

\title{Variations of the cop and robber game on graphs}
\author{
	\begin{tabular}{rcl}
	Carl Joshua Quines\thanks{equal first author contribution} & \qquad & Espen Slettnes\footnotemark[1] \\
	Shen-Fu Tsai & \qquad & Jesse Geneson\thanks{PSU math department}
    \end{tabular}
    \small}

\begin{document}
\maketitle
\begin{abstract}
  We prove new theoretical results about several variations of the cop and robber game on graphs. First, we consider a variation of the cop and robber game which is more symmetric called the cop and killer game. We prove for all $c < 1$ that almost all random graphs are stalemate for the cop and killer game, where each edge occurs with probability $p$ such that $\frac{1}{n^{c}} \le p \le 1-\frac{1}{n^{c}}$. We prove that a graph can be killer-win if and only if it has exactly $k\ge 3$ triangles or none at all. We prove that graphs with multiple cycles longer than triangles permit cop-win and killer-win graphs. For $\left(m,n\right)\neq\left(1,5\right)$ and $n\geq4$, we show that there are cop-win and killer-win graphs with $m$ $C_n$s. In addition, we identify game outcomes on specific graph products.

  Next, we find a generalized version of Dijkstra's algorithm that can be applied to find the minimal expected capture time and the minimal evasion probability for the cop and gambler game and other variations of graph pursuit.

  Finally, we consider a randomized version of the killer that is similar to the gambler. We use the generalization of Dijkstra's algorithm to find optimal strategies for pursuing the random killer. We prove that if $G$ is a connected graph with maximum degree $d$, then the cop can win with probability at least $\frac{\sqrt d}{1+\sqrt d}$ after learning the killer's distribution. In addition, we prove that this bound is tight only on the $\left(d+1\right)$-vertex star, where the killer takes the center with probability $\frac1{1+\sqrt d}$ and each of the other vertices with equal probabilities.
\end{abstract}

\section{Introduction}
The game of cop and robber on a graph is a simple model of the process of pursuing an adversary. Nowakowski and Winkler \cite{NW} and Quilliot \cite{Q} independently defined the game on a given graph $G$ and identified the graphs on which the cop has a winning strategy, assuming that both players use optimal strategies. Several other papers have studied different aspects of this game \cite{3,4,8,9,10}, such as the capture time.

In the original version of the game, the cop and robber play the following game on a graph $G$: the cop chooses a vertex, then the robber chooses a vertex, then the players take turns moving beginning with the cop. A move is either staying at one's present vertex or moving to an adjacent vertex, and both players see each move. The cop wins if they occupy the same vertex at some point, otherwise the robber wins.

Other versions of graph pursuit have also been studied, such as a variation where the robber becomes a rabbit and is able to hop between vertices \cite{1,2}. In another version of the game, the robber becomes a gambler and uses a fixed probability distribution on the vertices of the graph to determine its next vertex, moving simultaneously with the cop or cops \cite{geneson_arxiv_2016,geneson_arxiv_2017,komarov_arxiv_2013,tsai_arxiv_2017}.

In this report we prove new theoretical results about several variations of the cop and robber game. We consider a variation of the cop and robber game which is more symmetric. The two players are now a cop and a killer, also a perfect information game. The cop and killer each pick a vertex. Then the game begins, with the cop going first. She moves to an adjacent vertex, followed by the killer moving to an adjacent vertex. If after the cop moves she occupies the vertex where the killer is located, she wins. Similarly, if after the killer moves, he occupies the vertex where the cop is located, he wins. If neither happens after a infinite amount of turns, it is a stalemate.

We prove for all $c < 1$ that almost every random graph is stalemate for the cop and killer game, where each edge occurs with probability $\frac{1}{n^{c}} \le p \le 1-\frac{1}{n^{c}}$. Next, we prove that a graph can be killer-win if and only if it has either exactly $k\ge 3$ triangles  or none at all. We give examples of infinite graphs that are cop-win, killer-win, and stalemate. In addition, we prove that graphs with multiple cycles longer than triangles permit cop-win and killer-win graphs. For $\left(m,n\right)\neq\left(1,5\right)$ and $n\geq4$, there are cop-win and killer-win graphs with $m$ $C_n$s.

For connected graphs $H$ and $G$ with $|G|>1$ and $|H|>1,$ we find the result of the cop and killer game on several products of $G$ and $H$. We prove that $G~\square~ H$ is a stalemate if at least one of $G$ and $H$ is not a tree. We show that $G\times H$ is a stalemate if both $G$ and $H$ are not killer-win, and killer-win otherwise. Finally, we prove that $G\boxtimes H$ is cop-win if both $G$ and $H$ have a universal vertex, and is stalemate otherwise.

We find a generalized version of Dijkstra's algorithm that can be applied to find the minimal expected capture time and the minimal evasion probability for the cop and gambler game and other variations of pursuit on graphs. We further consider a randomized version of the killer that is similar to the gambler. We use this algorithmic method to find optimal strategies for pursuing the random killer. We prove that if $G$ is a connected graph with maximum degree $d$, then the cop can win with probability at least $\frac{\sqrt d}{1+\sqrt d}$ after learning the killer's distribution. In addition, we prove that this bound is tight only on the $\left(d+1\right)$-vertex star, where the killer takes the center with probability $\frac1{1+\sqrt d}$ and each of the other vertices with equal probabilities.

\section{Cop and killer}
As mentioned previously, this is a more symmetric variation of cop and robber. We found some basic properties of cop-win, killer-win, and stalemate graphs.
\begin{lemma}
  (1) All graphs with a universal vertex are cop-win.\\
  (2) All trees which are not stars are killer-win.\\
  (3) Triangle is cop-win, four-cycle is killer-win, and cycles with more than four vertices are stalemate graphs.\\
  (4) Grid graphs larger than $1\times3$ are killer-win. King's graphs larger than $3\times3$ are stalemate graphs.\\
  (5) A bipartite graph is cop-win if and only if it is a star.
\end{lemma}
\begin{proof}
  (1) is trivial.\\
  (2) If a tree has no universal vertex, then it has at least $3$ vertices. The killer could always pick a vertex that is $2$ steps away from the cop and then chase down the cop.\\
  (3) In a cycle with more than four vertices, whenever two players are two steps apart, the player who moves next moves away from the other and increases their shortest distance by one. Therefore they will never be one step apart, which is a stalemate.\\
  (4) In a $1\times k$ grid graph with $k>3$, the killer could pick a vertex two steps away from the cop. In a $k\times l$ grid graph with $k,l>1$, the killer could start at a vertex two steps away from but not horizontally or vertically aligned with the cop. In both cases the killer wins by keeping the same relative position to the cop. In a king's graph larger than $3\times3$, the killer can always choose an initial position at least two steps from the cop. Then, whenever they are exactly two steps apart, they must be vertically or horizontally aligned. Hence the next moving player can always move to a vertex further away from the other, therefore it is a stalemate.\\
  (5) A star is clearly cop-win. For a non-star bipartite graph $G$ with $V\left(G\right)=A\cup B$ and $A\cap B=\emptyset$, suppose cop starts at $v\in A$. If $|A|=1$ since $G$ is not a star, the killer can remain free forever by selecting a vertex in $B$ not adjacent to $v$; if $|A|>1$, then the killer starts at another vertex $u$ in $A$. He later either switches between $A$ and $B$ following the cop, or wins by capturing the cop.
\end{proof}

There are also interesting properties regarding retract and bipartiteness.
\begin{lemma} 
  \begin{enumerate}
  \item [(1)] Killer-win and stalemate graphs are all non-retract-closed.
  \item [(2)] If a graph is not stalemate, then it either has an universal vertex or there exist vertices $u, v\in V\left(G\right)$ such that $N\left(v\right)\subseteq N\left(u\right)$ and there is no edge between $u$ and $v$.
  \item [(3)] For each $m \geq 1$, there are non-bipartite killer-win graphs that are also $C_{2k+1}$-free for each $1 \leq k \leq m$.
  \item [(4)] The maximum number of edges is $\Theta\left(n^2\right)$ for a non-bipartite killer-win graph that is also $C_{2k+1}$-free for each $k \leq m$ with $m$ fixed.
  \end{enumerate}
\end{lemma}
\begin{proof}
  \begin{enumerate}
  \item [(1)] $C_4$ retracts to $P_3$, the former is killer-win and the latter is cop-win. Also $C_6$, stalemate, retracts to $P_4$, killer-win.
    
  \item [(2)] It suffices to show that a non-stalemate graph without universal vertex has two vertices $u$ and $v$ with $N\left(v\right)\subseteq N\left(u\right)$ and no edge between $u$ and $v$. Suppose in the last step a player moves from $u$ to $w$ and capture the other, and in the previous step the losing player moves from $v$ to $w$. If $N\left(v\right)$ is not a subset of $N\left(u\right)$, then the losing player could have moved to a vertex in $N\left(v\right)-N\left(u\right)$ without losing. If there is an edge between $u$ and $v$, then the losing player could have killed the winning player.
    
  \item [(3)] We construct a graph $G$ with vertices $A_1,A_2,\ldots,A_{2m+3},B_1,B_2,\ldots,B_{2m+3}$. $\left(A_i,B_{i+1}\right), \left(A_i,A_{i+1}\right), \left(B_i,B_{i+1}\right), \left(B_i,A_{i+1}\right)\in E\left(G\right)$ for each $i\in[2m+2]$. Also $\left(A_{2m+3},A_{1}\right),\left(B_{2m+3},A_{1}\right),\left(A_{2m+3},B_{1}\right),\left(B_{2m+3},B_{1}\right)\in E\left(G\right)$. Because $G$ has odd cycles it is not bipartite. Define clusters $K_i=\{A_i,B_i\}$ for $i\in[2m+3]$. The killer wins by starting at the vertex belonging to the same cluster $K_i$ as the cop. Moreover, clearly the smallest odd cycle has length $2m+3$. Figure \ref{fig:cyclefree} shows this for $m=5$.
    
    \begin{figure}[H]
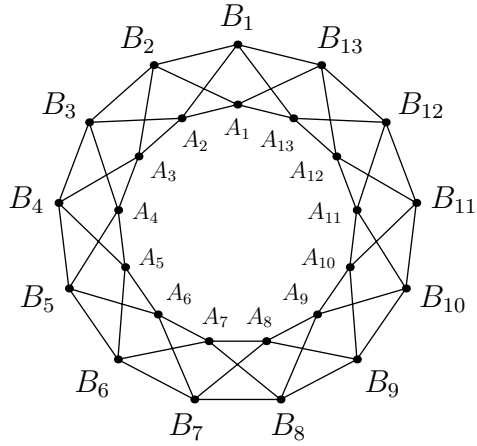

      \begin{asy}
        size(2.5inch);
        int rotations = 13;
        real orientation=90;
        for (int i=0; i<rotations; ++i)
        {
          dot(2*dir(orientation+i*360/rotations));
          label(scale(0.75)*Label("$A_{"+(string)(i+1)+"}$"), 2*dir(orientation+i*360/rotations), -1.5*dir(orientation+i*360/rotations));
          
          dot(3*dir(orientation+i*360/rotations));
          label("$B_{"+(string)(i+1)+"}$", 3*dir(orientation+i*360/rotations), 1.5*dir(orientation+i*360/rotations));
          
          draw(2*dir(orientation+i*360/rotations) -- 2*dir(orientation+(i+1)*360/rotations) -- 3*dir(orientation+i*360/rotations) -- 3*dir(orientation+(i+1)*360/rotations) -- cycle);
        }
        
      \end{asy}
      \centering
      \caption{\footnotesize A $C_{2k+1}$-free non-bipartite killer-win graph for each $1 \leq k \leq 5$.}
      \label{fig:cyclefree}
    \end{figure}
    
  \item [(4)] Similar to the construction in (3), assume $n\ge4m+6$ and build a graph with vertices labeled $1$ through $n,$ where vertices $u$ and $v$ are connected if and only if $u-v\equiv\pm1\pmod{2m+3}$. Define clusters $K_i=\{v \mid v \equiv i \pmod {2n+3} \}$ for $i \in [2m+3].$ Killer wins if he picks the same cluster as the cop. The graph has $\Theta\left(\left(\frac{n}{2m+3}\right)^2\right)$ edges and no $C_{2k+1}$ for each $k\in[m]$.
    
    \begin{figure}[H]
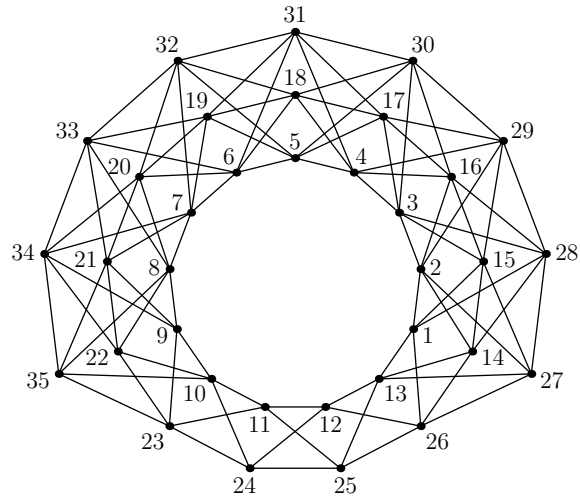

      \begin{asy}
        size(3inch);
        int rotations = 13;
        real orientation = 90-1440/13;
        int vertices = 35;
        int previousLayers = vertices # rotations;

        for (int i=0; i<vertices; ++i)
        {
          int layer = i # rotations;
          pair spot = (2+layer)*dir(orientation+i*360/rotations);

          for (int l=0; l<previousLayers; ++l)
          {
            draw((2+l)*dir(orientation+(i-1)*360/rotations) -- spot);
          }

          dot(spot);
          label(scale(0.75)*Label("$"+(string)(i+1)+"$"), spot, dir(orientation+i*360/rotations));

          previousLayers = vertices # rotations;
          if (i % rotations < vertices % rotations)
          ++previousLayers;
        }
      \end{asy}
      \centering
      \caption{\footnotesize A $C_{2k+1}$-free non-bipartite killer-win graph for each $1 \leq k \leq 5$ of size $\ge 4 \cdot 5 + 6.$}
      \label{fig:bigcyclefree}
    \end{figure}
  \end{enumerate}
\end{proof}

The next proof shows for all $c < 1$ that almost all random graphs are stalemate for the cop and killer game, where each edge occurs with probability $p$ such that $\frac{1}{n^{c}} < p < 1-\frac{1}{n^{c}}$.

\begin{lemma}
  Let $c<1.$ Then almost all random graphs are stalemate for the cop and killer game, where each edge occurs with probability $p$ such that $\frac{1}{n^c} \le p \le 1-\frac{1}{n^c}$.
\end{lemma}

\begin{proof}
  In the last lemma, we proved that if a graph is not stalemate, then it either has an universal vertex or there exist $u,v\in V\left(G\right)$ such that $N\left(v\right)\subseteq N\left(u\right)$ and there is no edge between $u$ and $v$. The probability of a universal vertex is at most $n p^{n-1}$ by the union bound, which has a limit of $0$ as $n \rightarrow \infty$. 

  The probability of having vertices $u,v\in V\left(G\right)$ such that $N\left(v\right)\subseteq N\left(u\right)$ and there is no edge between $u$ and $v$ is at most $n^2 \left(1+p\left(p-1\right)\right)^{n-2}$ by the union bound, which also has a limit of $0$ as $n \rightarrow \infty$. Thus if $s_{n}$ denotes the probability that the random graph on $n$ vertices is stalemate, then $\lim_{n \rightarrow \infty} s_{n} = 1$ for $\frac{1}{n^{c}} < p < 1-\frac{1}{n^{c}}$.
\end{proof}

Next we determine for each $k \geq 0$ if there are any killer-win graphs with $k$ triangles.

\begin{lemma}
  A graph can be killer-win if and only if it has no triangles or exactly $k\ge 3$ triangles.
\end{lemma}
\begin{proof}
  $C_4$ is a killer-win graph with no triangles.\\
  
  If there is exactly one triangle, the cop could pick a vertex in it and stay in the triangle. The moment the killer enters the triangle he loses. Moreover he is adjacent to at most one vertex, so the cop could always move to a vertex in the triangle not adjacent to the killer.\\

  If there are two triangles and they share a side, let it be $\left(u,v\right)$ and the other two vertices in the triangles are $x$ and $y$. Define $A=\{u,v,x,y\}$. The cop picks $u$. When the killer is outside of the triangles, he is adjacent to at most a vertex in $A$, so the cop could always move to another vertex in $A$. The only possible way for the killer to enter $A$ is to move to $x$ when the cop is at $y$ or vice versa. Clearly, the cop could have moved between $u$ and $v$ instead, and thus the killer would not move to $x$ or $y$.\\

  If the two triangles do not share a side, then similar to the one-triangle situation, the cop could always move to a vertex in the same triangle not adjacent to the killer.\\

  We provide a killer-win graph with three triangles. Consider a pentagon and an extra vertex $x$. Exactly four of the vertices in the pentagon are adjacent to $x,$ as in Figure \ref{fig:three3s}. It is not hard to verify that this is killer-win.\\
  \begin{figure}[H]
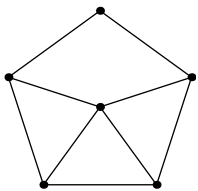

    \begin{asy}
      size(1inch);
      //string names[]={"$v$","$u$","$y$","$z$","$w$"};
      for (int i=0; i<5; ++i)
      {
        dot(dir(90+72*i));
        draw(dir(90+72*i)--dir(162+72*i));
        if (i>0)
        {
          draw((0,0)--dir(90+72*i));
        }
        //label(names[i],dir(90+72*i),dir(90+72*i));
      }
      dot((0,0));
      //label("$x$",(0,0),N);
    \end{asy}
    \centering
    \caption{\footnotesize A killer-win graph with three triangles.}
    \label{fig:three3s}
  \end{figure}

  Finally, for $k>3$, we can construct $k$ triangles $T_1,\ldots,T_k$, such that
  \begin{itemize}[noitemsep,topsep=0pt]
    \item [$\bullet$] $T_{i}$ and $T_{i+1}$ share a side for all $1 \le i < k,$
    \item [$\bullet$] there are no more common sides, and 
    \item [$\bullet$] there is no common vertex to all $k$ triangles.   
  \end{itemize}

    We then 3-color this graph, as shown in Figures $\ref{fig:many3s}$ and $\ref{fig:worm}.$ The killer can always pick a vertex of the same color as the cop. In subsequent steps the killer should move towards the cop, but matching the cop's color at every step.  The cop will eventually be cornered at $T_1$ or $T_k.$
\end{proof}
\begin{figure}[H]
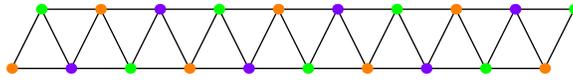

  \centering
  \begin{asy}
    size(3inch);
    int n = 20;
    pen colors[] = { orange, green, purple };
    
    for (int i = 0; i < n; ++i)
    {
      if (i < n-2)
      draw(((i+1)/2.0, (i+1)%2) -- (i/2.0, i%2) -- (i/2.0+1, i%2));
      if (i == n-2)
      draw(((i+1)/2.0, (i+1)%2) -- (i/2.0, i%2));
      dot((i/2.0, i%2), colors[i % colors.length]+4);
    }
  \end{asy}
  \caption{\footnotesize A killer-win graph with more than 3 triangles.}
  \label{fig:many3s}
\end{figure}
\begin{figure}[H]
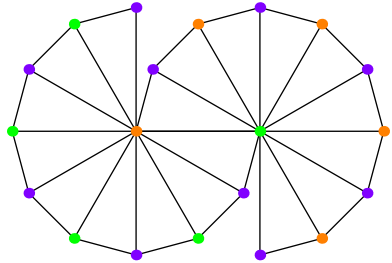

  \centering
  \begin{asy}
    size (2inch);
    pair A = (0,0);
    pair B = (1,0);
    pen Acolors[] = { green, purple };
    pen Bcolors[] = { orange, purple };
    //pen Acolors[] = { black };
    //pen Bcolors[] = { black };
    
    void pacmanlines (pair center, real start, real n, real increment, real radius=1.0)
    {
      pair previous = (0,0);
      for (int i=0; i<n; ++i)
      {
        pair edge = center + radius*expi(start+i*increment);
        draw(center--edge);
        if (i>0)
        draw(previous -- edge);
        previous = edge;
      }
    }
    
    void pacmandots (pair center, real start, real n, real increment, pen colors[], real radius=1.0)
    {
      for (int i=0; i<n; ++i)
      {
        pair edge = center + radius*expi(start+i*increment);
        pen color = colors[i % colors.length];
          dot(edge, color+4bp);
        }
      }
      
      pacmanlines(A, 0, 10, -pi/6);
      pacmanlines(B, -pi, 10, -pi/6);
      pacmandots(A, 0, 10, -pi/6, Acolors);
      pacmandots(B, -pi, 10, -pi/6, Bcolors);
      
      //label("C", (0,1), N);
      //label("K", (1,0), dir(220));
  \end{asy}
  \caption{\footnotesize Another killer-win graph with more than 3 triangles.}
  \label{fig:worm}
\end{figure}
From previous results we have:
\begin{lemma}
  There are infinite graphs that are cop-win, killer-win, or stalemate.
\end{lemma}
\begin{proof}
  Stars are cop-win, paths with more than three vertices are killer-win, and cycles with more than four vertices are stalemate, as seen in Figure \ref{fig:infinte}.
  \begin{figure}[H]
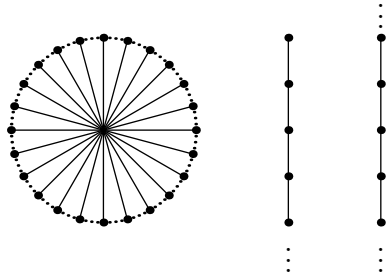

    \centering
    \begin{asy}
      size(2inch);
      int steps = 12;
      for (int i = 0; i < 2*steps; ++i)
      {
        draw((0,0)--expi(i*pi/steps));
        pen dotpen = currentpen;
        for (int d=0; d<4; ++d)
        {
          dot(expi((i+d/4)*pi/steps), dotpen);
          dotpen=currentpen+1bp;
        }
      }
      draw((2,-1)--(2,1));
      draw((3,-1)--(3,1));
      label("$\vdots$",(3,-1),S); label("$\vdots$",(3,1),N);
      label("$\vdots$",(2,-1),S);
      
      int num=4;
      for (real y=-1; y<=1; y+=2/num)
      {
        dot((2,y));
        dot((3,y));
      }
    \end{asy}
    \caption{\footnotesize Infinite graphs that are cop-win, killer-win, and stalemate.}
    \label{fig:infinte}
  \end{figure}
  
\end{proof}

Unlike triangles, longer cycles permit cop-win and killer-win graphs.
\begin{lemma}
  For $\left(m,n\right)\neq\left(1,5\right)$ and $n\geq4$, there are cop-win and killer-win graphs with $m$ $C_n$s.
\end{lemma}
\begin{proof}
  To make a cop-win graph with $m$ $C_n$s, consider $m$ $C_n$s. Each pair of $C_n$s share a fixed vertex $v$, and $v$ is also the universal vertex within each $C_n,$ as in Figure \ref{fig:petal}. The cop picks $v$ and wins.
  
  \begin{figure}[H]
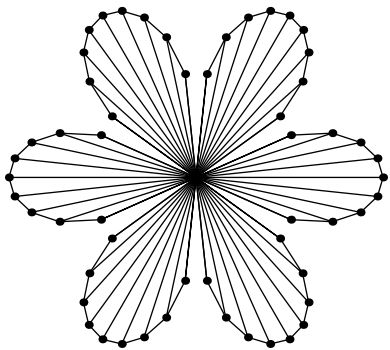

    \centering
    \begin{asy}
      size(2inch);
      
      pair f (real theta)
      { 
        //return (cos(4*theta) + cos(2*theta), sin(4*theta) - sin(2*theta));
        return (sqrt(abs(cos(3*theta))) * cos(theta),
        sqrt(abs(cos(3*theta))) * sin(theta));
      }
      for (int i=0; i<30; ++i)
      {
        pair a = f((i+1)*pi/30);
        pair b = f(i*pi/30);
        pair c = -b;
        pair d = -a;
        draw(a--b--c--d);
        dot(b);
        dot(c);
      }
      
    \end{asy}
    \caption{\footnotesize A cop-win graph with multiple cycles of length $n\ge4$}
    \label{fig:petal}
  \end{figure}

  We now construct a killer-win graph with $m$ $C_n$s. 
  \begin{enumerate}
  \item If $n=4$, take $n$ disconnected copies of $C_4$ as shown in Figure \ref{fig:many4s}.
    \begin{figure}[H]
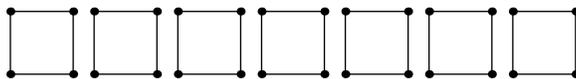

      \centering
      \begin{asy}
        size(3inch);
        int m = 7;
        //pen colors[] = { orange, green, purple };
        void fourgon (pair position, real sidelength=0.75)
        {
          pair vertices[] = { position, position+(sidelength,0), 
            position+(sidelength,sidelength), position+(0, sidelength) };
          draw(vertices[0]--vertices[1]--vertices[2]--vertices[3]--cycle);
          for (int v=0; v<4; ++v)
          dot(vertices[v]);
        }
        
        for (int i=0; i<m; ++i)
        {
          fourgon((i,0));
        }
      \end{asy}
      \caption{\footnotesize A killer-win graph with multiple cycles of length $n=4$}
      \label{fig:many4s}
    \end{figure}

  \item If $n>4$, take triangles $T_1,\ldots,T_{m+n-3}$, such that
    \begin{itemize}[noitemsep,topsep=0pt]
    \item [$\bullet$] $T_{i}$ and $T_{i+1}$ share a side for all $1 \le i < m+n-3,$
    \item [$\bullet$] there are no more common sides, and 
    \item [$\bullet$] there is no common vertex to all $(m+n-3)$ triangles.     
    \end{itemize}
    We then 3-color this graph, as shown in Figure $\ref{fig:manycycles}.$ The killer can always pick a vertex of the same color as the cop. In subsequent steps the killer should move towards the cop, but matching the cop's color at every step. The cop will eventually be cornered.
    \begin{figure}[H]
      \centering
      \begin{asy}
        size(3inch);
        int n = 24;
        pen colors[] = { orange, green, purple };
        
        for (int i = 0; i < n; ++i)
        {
          if (i < n-2)
          draw(((i+1)/2.0, (i+1)%2) -- (i/2.0, i%2) -- (i/2.0+1, i%2));
          if (i == n-2)
          draw(((i+1)/2.0, (i+1)%2) -- (i/2.0, i%2));
          dot((i/2.0, i%2), colors[i % colors.length]+4);
        }
      \end{asy}
      \caption{\footnotesize A killer-win graph with multiple cycles of length $n\ge5$}
      \label{fig:manycycles}
    \end{figure}
  \end{enumerate}
\end{proof}
The Cartesian product of two graphs turns out to be a stalemate graph in most cases.
\begin{lemma}
  Let $G$ and $H$ be connected graphs such that $|H|>1$ and $G$ is not a tree, then $G\square H$ is stalemate.
\end{lemma}
\begin{proof}
  We use vertex pair $\left(u,v\right)$ to represent a vertex in $G\square H$ for $u\in V\left(G\right), v\in V\left(H\right)$. For distinct $u_1,u_2\in V\left(G\right)$, $v_1,v_2\in V\left(H\right)$, $\left(u_1,v_1\right)$ is not adjacent to $\left(u_2,v_2\right)$ in $G\square H$, which therefore has no universal vertex.\\

  Then it suffices to show that given two players that are not adjacent to each other, the next player could always move to another vertex such that they are still not adjacent. We prove a stronger condition, that both players can stay in a cycle of length $3$ or more in $G$ forever. Let the players be at $\left(u_1,v_1\right)$ and $\left(u_2,v_2\right)$, and assume that the former player moves next. If $u_1=u_2$, then $v_1$ and $v_2$ are not equal or adjacent, so she could move to $\left(u_3,v_1\right)$ for any $u_3$ in the cycle in $G$ that is adjacent to $u_1$. If $u_1$ is adjacent to $u_2$ in $G$, then $v_1 \neq v_2$. Thus the former player can still move to $\left(u_3,v_1\right)$ where $u_3$ is in the same cycle in $G$ and $u_3 \neq u_2$. If $u_1$ is not adjacent to $u_2$, then she can move to $\left(u_{1},v'\right)$ where $v'$ may be identical to $v_1$.
\end{proof}

We find a similar kind of result for tensor products.

\begin{lemma}
  For any two graphs $|G|>1$ and $|H|>1,$ $G\times H$ is a stalemate if and only if both are not killer-win, and is killer-win if and only if at least one is killer-win.
\end{lemma}

\begin{proof}
  It is helpful to think of the tensor product of two graphs as two separate games being played simultaneously. For the first claim, let the cop choose $\left(g,h\right)\in G\times H$ to start. The killer cannot lose by choosing $\left(g,h'\right)\in G\times H$ and matching the cop's first coordinate every move thereafter. The cop cannot lose by playing optimally on both coordinates, as the killer cannot match both coordinates.

  For the second claim, without loss of generality suppose that $H$ is killer-win. Let the cop choose $\left(g,h\right)\in G\times H$ to start. The killer wins by choosing $\left(g,h'\right)\in G\times H$ where $h'$ is the optimal response to $h$ in $H,$ and matching the cop's first coordinate every move thereafter while playing optimally on the second coordinate.
\end{proof}

As with Cartesian products and tensor products, we can also determine the result of the cop and killer game for strong products based on the graphs in the product.

\begin{lemma}
  Let $G$ and $H$ be connected graphs where $G$ has no universal vertices and $|H| > 1.$ Then $G\boxtimes H$ is a stalemate.
\end{lemma}

\begin{proof}~
	\begin{itemize}[noitemsep,topsep=0pt] 
	  	\item There are no universal vertices, because for any vertex $(g,h) \in V\left(G\boxtimes H\right)$ there is a vertex $g' \in V\left(G\right)$ such that $g'$ is not adjacent to $g,$ so $(g',h)$ is not adjacent to $(g,h).$
	  	
	  	\item There exist no vertices $u, v\in V\left(G\boxtimes H\right)$ where $N\left(v\right)\subseteq N\left(u\right)$ with no edge between $u$ and $v.$  For any non-adjacent vertices $(g,h)$ and $(g',h'),$ 
	  	$$\left( g \not\sim g' \wedge g \neq g' \right) \vee \left( h \not\sim h' \wedge h \neq h' \right).$$
	  	\begin{itemize}
	  		\item If $\left( g \not\sim g' \wedge g \neq g' \right),$ then for any $h_1 \sim h,$ $N(g,h) \ni (g,h_1) \not\in N(g',h').$
	  		\item If $\left( h \not\sim h' \wedge h \neq h' \right),$ then for any $g_1 \sim g,$ $N(g,h) \ni (g_1,h) \not\in N(g',h').$
	  	\end{itemize}
	\end{itemize}
	Thus by lemma 2, this graph is a stalemate.
\end{proof}
Combining the results about graph products, we obtain the following corollary:

\begin{corollary}
  Let $G$ and $H$ be connected graphs with $|G|>1$ and $|H|>1.$
  \begin{enumerate}
  \item $G~\square~ H$ is a stalemate if at least one of $G$ and $H$ is not a tree.
  \item $G\times H$ is a stalemate if both $G$ and $H$ are not killer-win, and killer-win otherwise.
  \item $G\boxtimes H$ is cop-win if both $G$ and $H$ have a universal vertex, and is stalemate otherwise.
  \end{enumerate}
\end{corollary}

%\begin{corollary}
%  The following sets of graphs are closed under the specified graph products:
%  \begin{enumerate}
%  \item Connected stalemate graphs are closed under the Cartesian product.
%  \item The union of stalemate and killer-win graphs is closed under the Cartesian product.
%  \item Stalemate graphs are closed under the tensor product.
%  \item Killer-win graphs are closed under the tensor product.
%  \item The union of stalemate and killer-win graphs is closed under the tensor product.
%  \item The union of cop-win and stalemate graphs is closed under the tensor product.
%  \item Stalemate graphs are closed under the strong product.
%  \item The union of cop-win and stalemate graphs is closed under the strong product.
%  \item The union of stalemate and killer-win graphs is closed under the strong product.
%  \end{enumerate}
%\end{corollary}

\section{Algorithmic}
\label{section:algorithmic}
In \cite{tsai_arxiv_2017}, the author explores algorithms to compute the optimal cop strategy with random initial vertex and non-optimal robber distribution in the cop vs gambler game. The problem is analogous to Single Source Shortest Path Problem \cite{sssp}. Dijkstra's algorithm \cite{dijkstra} and Bellman-Ford algorithm \cite{bellman_ford} are adapted to find the optimal chase path of the cop that minimizes $T\left(v\right)$, the expected capture time of the cop starting from vertex $v$. The essence of these algorithms is to keep updating $t\left(v\right)$, the minimum expected capture time computed so far, by
$$
t\left(v\right)=\min\left(t\left(v\right), 1+\left(1-p_v\right)t\left(u\right)\right)
$$ where $u$ is adjacent to $v$.\\

It is possible to extend the update framework to minimize more objective functions. If the function $T\left(v\right)$ to be minimized satisfies
$$
T\left(v\right)=\min\left(J\left(v\right), \min_{u\in N\left(v\right)}H\left(v,u\right)\right)
$$ where $J\left(v\right)$ serves as upper bound of $T\left(v\right)$ and $N\left(v\right)$ is the set of vertices adjacent to $v$, then the adapted algorithms in \cite{tsai_arxiv_2017} compute $T\left(v\right)$ correctly. Of course, similar to the Single Source Shortest Path Problem, there are additional constraints: for generalized Dijkstra's algorithm to work, for every pair of adjacent $u$ and $v$, $T\left(v\right)\geq T\left(u\right)$ if $T\left(v\right)=H\left(v,u\right)$; for generalized Bellman-Ford algorithm to work, there should not be any cycle $C=\left(v_1,v_2,\ldots, v_k=v_1\right)$ such that for each $i\in[k-1]$, $H\left(v_i,v_{i+1}\right)=T\left(v_i\right)<T\left(v_{i+1}\right)$ because otherwise $T\left(v\right)$ is not well defined.\\

To find the minimum expected capture time for the cop and the gambler, $J\left(v\right)=1/p_v$ and $H\left(v,u\right)=1+\left(1-p_v\right)T\left(u\right)$. To generalize, when the cop has to spend $n\left(u,v\right)$ turns crossing edge $\left(u,v\right)$ from $u$ to $v$, $J\left(v\right)=1/p_v$ and $H\left(v,u\right)=1+\left(1-p_v\right)\left(T\left(u\right)+n\left(v,u\right)\right)$.\\

Let $e\left(m,v\right)$ be the minimum evasion probability for the gambler in $m$ turns when the cop starts at vertex $v$, with $e\left(m,v\right)=1$ when $m\leq0$. To compute $e\left(m,v\right)$, $J\left(v\right)=\left(1-p_v\right)^m$ and $H\left(v,u\right)=\left(1-p_v\right)e\left(m-1-n\left(v,u\right),u\right)$.\\

If the cop has to capture the gambler in $m$ turns, and the gambler occupies vertex $v$ with probability $p_{v,i}$ when the cop still has $i$ rounds left, then to compute $e\left(m,v\right)$ we need 
\begin{align*}
  J\left(v\right)&=\prod_{i=1}^m\left(1-p_{v,i}\right) \text{~~~and} \\
  H\left(v,u\right)&=\left(1-p_{v,m}\right)e\left(m-1-n\left(v,u\right),u\right).
\end{align*}

We explore a variation where more than one cop chases the gambler, which was defined in \cite{geneson_arxiv_2016}. In each round, each of the $k$ cops independently moves to an adjacent vertex or stays where they are.

If cops start at random initial positions before they know gambler's distribution, then the algorithms in \cite{tsai_arxiv_2017} could be extended to find the optimal strategy of the cops to minimize expected capture time. The key is to treat the joint positions of $k$ cops $\left(v_1,\ldots,v_k\right)$ as a vertex in the supergraph $G'=\left(V\left(G'\right),E\left(G'\right)\right)=\left(V\left(G\right)^k,E'\right)$, where $\left(\left(v_1,\ldots,v_k\right),\left(u_1,\ldots,u_k\right)\right)\in E'$ if and only if $v_i$ is adjacent or equal to $u_i$ for each $i\in[k]$.

\begin{definition}
  $\delta_i\left(U\right)=1$ if $v_i\in U$, otherwise $\delta_i\left(U\right)=0$.
\end{definition}

Using the notation from the beginning of this section,
\begin{align*}
  J\left(\{u_1,\ldots,u_k\}\right)&=\frac{1}{\sum_{i=1}^n\delta_i\left(\{u_1,\ldots,u_k\}\right)p_i}, \text{~~~and} \\
  H\left(U,U'\right)&=1+\left(1-\sum_{i=1}^n\delta_i\left(U\right)p_i\right)T\left(U'\right).
\end{align*}

The supergraph has $\left|V\left(G'\right)\right|=\left|V\left(G\right)\right|^k$ vertices and $\left|E\left(G'\right)\right|=O\left(\left(\left|V\left(G\right)\right|+\left|E\left(G\right)\right|\right)^k\right)$ edges.

\section{Cop and random killer}
Consider the following graph pursuit game that combines elements of cop and killer and cop and gambler.\\

For any given connected graph $G$, the cop chooses any starting vertex of $G$ and the random killer chooses a probability distribution $\{p_v\}_{v\in V\left(G\right)}$ on the vertices of $G$. The cop does not know the distribution, but the random killer knows the cop's starting position before they choose the distribution. On the odd turns, the cop can either move to an adjacent vertex or stay at their current vertex. On the even turns, the random killer hops to a vertex using their probability distribution.\\

The random killer wins if they land on the cop's current vertex, and the cop wins if they land on the killer's current vertex. They could end up with stalemate, for example if the cop stays at a vertex $v$ such that $p_v+\sum_{u\in N\left(v\right)}p_u=0$. If both cop and killer play to maximize their probability of victory, we say that $G$ is cop-win if the cop has higher probability of victory, and otherwise $G$ is random killer-win.\\

\begin{lemma}
  For the cop and known random killer, an optimal killer will not assign a probability of $0$ to the cop's initial vertex.
\end{lemma}
\begin{proof}
  If the cop starts from a vertex $v_1$ with $0$ probability of being visited by the killer, then there is a path $v_1,v_2,\ldots,v_k$ such that $p_{v_i}=0$ for $i\in[k-1]$ and $p_{v_k}>0$. The cop could just follow this path and wait at $v_{k-1}$ until the killer takes $v_k$, then the cop wins. Note that the cop cannot lose because she takes this path with zero probability of being landed by the killer.
\end{proof}

The algorithm described in Section \ref{section:algorithmic} could be applied to decide the cop's strategy that maximizes her winning probability. In particular, to maximize winning probability $T\left(v\right)$ of the cop at vertex $v$ when it is killer that moves next, we have
$$
T\left(v\right)=\max\left(J\left(v\right), \max_{u\in N\left(v\right)}H\left(v,u\right)\right),
$$
$H\left(v,u\right)$ and $J\left(v\right)$ are required.
The probability of cop winning by staying at $v$ is $J\left(v\right)=\frac{\sum_{u\in N\left(v\right)}p_u}{p_v+\sum_{u\in N\left(v\right)}p_u}$. The update candidate of winning probability by moving from $v$ to $u$ is $\sum_{w\in N\left(v\right)}p_w$, the chance that the killer next lands on $N\left(v\right)$ before the cop moves, plus $T\left(u\right)[1-p_v-\sum_{w\in N\left(v\right)}p_w]$. So
$$
H\left(v,u\right)=T\left(u\right)\left[1-p_v-\sum_{w\in N\left(v\right)}p_w\right]+\sum_{w\in N\left(v\right)}p_w.
$$
Notice here that maximizing the cop's winning probability may also increase losing probability. For example, a cop staying at a vertex $v$ with $p_v+\sum_{u\in N\left(v\right)}p_u=0$ never wins or loses, yet moving to other vertices may make the chances of winning and losing both non-zero.\\

Interestingly the cop almost always has a higher chance to win than the killer as long as the killer's distribution is known.
\begin{lemma}
  If the killer's distribution is known to the cop, then the only two graphs for the killer to have an equal or higher chance than the cop of winning are $P_1$ and $P_2$.
\end{lemma}
\begin{proof}
  We denote $n_v=\sum_{u\in N\left(v\right)}p_u$ and let the cop start at $v$. The only possible graphs that the cop does not have higher winning chance than the killer are those that in the first round the cop could not move to or stay at a vertex $u$ such that $n_u>p_u$, which could only happen if $n_v\leq p_v$. If $p_v=n_v$, then we still need $v$ to have only one neighbor $u$, $p_v=p_u$, and $u$ also has degree $1$, which is only possible for $P_2$. If $p_v>n_v$, then unless $v$ is neighborless the cop could always move to a neighbor $u$, then $n_u\geq p_v>n_v\geq p_u$ and the cop is more likely to win. So the only other possibility is $P_1$. Note that for a $P_2$ to make killer at least as likely to win as the cop his distribution must be $\left(0.5,0.5\right)$.
\end{proof}

We further quantify the cop's advantage when learning the killer's distribution before her first move by proving a lower bound on her winning probability in terms of the graph's maximum degree.
\begin{lemma}
  \label{lemma:star}
  If connected $G$ has maximum degree $d$, then cop can win with probability at least $\frac{\sqrt d}{1+\sqrt d}$ after learning killer's distribution.
\end{lemma}
\begin{proof}
  The cop could start at vertex $v$ with degree $d$. After killer's distribution is revealed, she decides to stay or move by maximizing $n_u/p_u$ where $u\in v\cup N\left(v\right)$. Consider $u=\arg\min_{w\in N\left(v\right)}p_w$, then $p_u\leq n_v/d$ and $n_u/p_u\geq dp_v/n_v$. Therefore the product of quantities $n_v/p_v$ and $n_u/p_u$ is at least $d$, making the larger of the two at least $\sqrt{d}$.
\end{proof}

It is noteworthy that there is connected graph with maximum degree $d$ that does not allow cop a winning chance higher than $\frac{\sqrt d}{1+\sqrt d}$: for a star with $d+1$ vertices, the killer can assign a probability of $\frac1{1+\sqrt d}$ to the center vertex, and an equal probability to the rest of the vertices. Moreover, it is the only graph with maximum degree $d$ that the cop has chance of winning as low as $\frac{\sqrt d}{1+\sqrt d}$.

\begin{corollary}
  This bound is tight only on the $\left(d+1\right)$-vertex star, where the killer takes the center with probability $\frac1{1+\sqrt d}$ and each of the other vertices with equal probabilities.
\end{corollary}
\begin{proof}
  We examine the inequality in the proof of Lemma \ref{lemma:star}, and observe that the equality holds only when the initial vertex $v$ (with degree $d$) has chance $\frac{1}{1+\sqrt{d}}$, each of its neighbors has chance $\frac{1}{\sqrt{d}\left(\sqrt{d}+1\right)}$, and none of them have any neighbor with non-zero probability. It remains to show that none of them have any other neighbor. Clearly $v$ does not have any other neighbor because it has degree $d$. For each of $v$'s neighbor $u$, if it has a $0-$chance neighbor $w$, then the cop could move from $v$ to $u$ in the first turn, and then from $u$ to $w$ in the second and stay there forever. In this way her chance of winning is
  $$
  \frac{1}{\sqrt{d}+1}+\left(1-\frac{1}{\sqrt{d}\left(\sqrt{d}+1\right)}-\frac{1}{\sqrt{d}+1}\right).
  $$
  It is not hard to see that this winning probability is greater than or equal to $\frac{\sqrt{d}}{\sqrt{d}+1}$. The equality holds only when $d=1$, in which case $w$ does not exist because $u$ already has a neighbor $v$.
\end{proof}

\section{Acknowledgments}
This paper has resulted from the 2017 CrowdMath project on graph algorithms and applications (online at \url{http://www.aops.com/polymath/mitprimes2017a}). CrowdMath is an open program created jointly by the MIT Program for Research in Math, Engineering, and Science (PRIMES) and the Art of Problem Solving that gives students all over the world the opportunity to collaborate on a research project.

The contributors for this paper were Espen Slettnes (goodbear), Carl Joshua Quines (cjquines0), Shen-Fu Tsai (parityhome), and Jesse Geneson (JGeneson).

\end{document}